\documentclass[twocolumn,aps,pra,superscriptaddress,tightlines,nofootinbib]{revtex4-2}
\usepackage{hyperref}
\hypersetup{
colorlinks=true,
linkcolor=blue,         
citecolor=blue,         
filecolor=blue,      
urlcolor=blue
}
\usepackage{url}
\usepackage{float}
\usepackage{placeins}
\usepackage{color}
\usepackage{graphicx,epsfig,colordvi,epstopdf}
\usepackage{bm}
\usepackage{mathtools}
\usepackage{mathdots}
\usepackage{physics}

\usepackage{amsmath}
\usepackage{amsfonts}
\usepackage{amssymb}
\usepackage{amsthm}
\usepackage{cleveref}
\newcommand\numberthis{\addtocounter{equation}{1}\tag{\theequation}}

\newtheorem*{theorem}{Theorem}
\newtheorem*{proposition}{Proposition}
\newtheorem*{definition}{Definition}

\begin{document}
\title{Machine learning phase transitions: Connections to the Fisher information}

\author{Julian Arnold}
\affiliation{Department of Physics, University of Basel, Klingelbergstrasse 82, 4056 Basel, Switzerland}
\author{Niels Lörch}
\affiliation{Department of Physics, University of Basel, Klingelbergstrasse 82, 4056 Basel, Switzerland}
\author{Flemming Holtorf}
\affiliation{Department of Chemical Engineering, Massachusetts Institute of Technology, Cambridge, MA 02139, USA}
\affiliation{CSAIL, Massachusetts Institute of Technology, Cambridge, MA 02139, USA}
\author{Frank Sch\"afer}
\affiliation{CSAIL, Massachusetts Institute of Technology, Cambridge, MA 02139, USA}
\date{\today}

\begin{abstract}
Despite the widespread use and success of machine-learning techniques for detecting phase transitions from data, their working principle and fundamental limits remain elusive. Here, we explain the inner workings and identify potential failure modes of these techniques by rooting popular machine-learning indicators of phase transitions in information-theoretic concepts. Using tools from information geometry, we prove that several machine-learning indicators of phase transitions approximate the square root of the system's (quantum) Fisher information from below -- a quantity that is known to indicate phase transitions but is often difficult to compute from data. We numerically demonstrate the quality of these bounds for phase transitions in classical and quantum systems.
\end{abstract}

\maketitle
\emph{Introduction}.---Traditionally, critical phenomena have been studied by relying on the Ginzburg-Landau-Wilson paradigm which is based on concepts such as symmetry breaking and local order parameters~\cite{goldenfeld:2018}. This framework fails to describe topological phase transitions~\cite{berezinskii:1971,kosterlitz:1973} for which there is no local order parameter. Moreover, identifying the proper order parameters of systems whose symmetry-breaking patterns are unknown is difficult. Information-theoretic quantities are particularly promising for studying phase transitions without relying on this traditional paradigm. Such quantities are universal and their computation does not require a detailed analysis of the system's physics, such as its order parameters.

In this context, the classical Fisher information (FI)~\cite{fisher:1922} and its quantum counterpart~\cite{helstrom:1967} have been extensively studied as universal indicators of phase transitions, i.e., as quantities whose maxima are signatures of critical points. The FI is a generalized susceptibility that measures the sensitivity of the system with respect to a tuning parameter. In the case of classical equilibrium systems, it measures fluctuations in the system's collective variables and is proportional to well-known response functions, such as the magnetic susceptibility or the heat capacity~\cite{prokopenko:2011}. Similarly, the quantum FI reduces to the fidelity susceptibility~\cite{you:2007,gu:2010,liu:2014}, i.e., the leading-order response of the fidelity between quantum states to parameter fluctuations. The fidelity susceptibility has been shown to detect symmetry-breaking~\cite{venuti:2007,zanardi:2007}, topological~\cite{abasto:2008,yang:2008,garnerone:2009}, and Berezinskii-Kosterlitz-Thouless-type (BKT-type)~\cite{yang:2007,wang:2010} quantum phase transitions. Moreover, the quantum FI has been used to investigate finite-temperature transitions as well as non-equilibrium phenomena, such as dissipative~\cite{banchi:2014,rota:2017,heugel:2019}, dynamical~\cite{macieszczak:2016,guan:2021}, or excited-state~\cite{zhou:2023} phase transitions.

Recently, also machine learning (ML) has emerged as an alternative paradigm for studying phase transitions~\cite{carleo:2019, carrasquilla:2020, dawid:2022}. The appeal of ML methods is akin to the one of information-theoretic approaches: they are generic and can be used to characterize a system using minimal explicit knowledge of its underlying physics. A large class of ML methods are based on solving classification or regression tasks using predictive models such as neural networks (NNs)~\cite{carrasquilla:2017,van:2017,schaefer:2019,arnold:2023,guo:2023}. By analyzing the model predictions, indicators of phase transitions are computed whose local maxima mark critical points. This framework has been employed to investigate many systems, including symmetry-breaking~\cite{carrasquilla:2017,van:2017,chng:2017,suchsland:2018,liu:2018,lee:2019,rem:2019,schaefer:2019,arnold:2021,arnold:2022,miles:2023,guo:2023,schlomer:2023,arnold:2023}, topological~\cite{carrasquilla:2017,van:2017,beach:2018,suchsland:2018,liu:2018,lee:2019,rem:2019,greplova:2020,arnold:2022,arnold:2023,guo:2023}, and non-equilibrium~\cite{van:2017,schaefer:2019,guo:2020,bohrdt:2021,zvyagintseva:2021,arnold:2022,guo2:2023} phase transitions in both the classical and quantum realm. Despite the rampant use and success of these heuristic ML-based methods, their working principle and fundamental limits have remained elusive.

\indent In this Letter, we put these methods on a firm theoretical footing by rooting them in information theory. In particular, we
prove that a large class of ML indicators for phase transitions are lower bounds to the square root of the system's FI with respect to the tuning parameter. These bounds reveal a strong link between the ML and information-theoretic paradigm for detecting phase transitions. We numerically demonstrate the quality of these underapproximations for phase transitions in classical and quantum systems. Building upon previous results on the FI in the context of statistical and quantum physics, this yields insights into the operation of ML methods for detecting phase transitions. These insights help to understand the limitations and strengths of such methods when applied to different classes of phase transitions, suggest the methods' usage as algorithms for approximating the Fisher information in more general settings, and improve their performance via modifications motivated by information geometric correspondences.

\emph{Detecting phase transitions from data}.---We study a physical system with a discrete state space $\mathcal{X}$ and, for simplicity, we assume it to be characterized by a single tunable parameter $\gamma$ along which a phase transition occurs. Generalizations of our results to higher-dimensional parameter spaces featuring multiple phase transitions are discussed in Ref.~\cite{suppl}. The probability of measuring the system in state $\bm{x} \in \mathcal{X}$ at a given tuning parameter realization $\gamma$ is denoted by $P(\bm{x}|\gamma)$. The goal is to detect the critical point based on measurements (samples) of the system state at a discrete set of realizations of the tuning parameter $\Gamma$. In the following, we discuss three distinct ML approaches to detect phase transitions from such data.

\indent The first approach is based on solving a classification task~\cite{carrasquilla:2017,arnold:2023}. To that end, let us assume we know
a set of points $\Gamma_{0}$ and $\Gamma_{1}$ lying within each of the two phases, phase 0 and phase 1, where $\forall \gamma \in \Gamma_{0}, \gamma' \in \Gamma_{1}: \gamma < \gamma'$. Then, for $y \in \left\{0, 1\right\}$ we can assign to all the samples $\bm{x}$ drawn at points in $\Gamma_{y}$ the label $y(\bm{x}) = y$. Based on this, we train a classifier $\hat{y}: \mathcal{X} \rightarrow [0,1]$ to assign the correct phase to a given sample $\bm{x}$. Intuitively, the mean prediction $\hat{y}(\gamma) = \mathbb{E}_{\bm{x} \sim P(\cdot|\gamma)}\left[ \hat{y}(\bm{x})\right]$ will change most at the critical point. This is captured by the following scalar indicator of phase transitions
\begin{equation}\label{eq:SL_ind}
    I_{1}(\gamma) = \left| \frac{\partial \hat{y}(\gamma)}{\partial \gamma} \right|,
\end{equation}
whose maximum is expected to occur at the critical point.

\indent For this first approach to be effective, one requires partial knowledge of the phase diagram, which may be unavailable. To get around this, one can employ a phase-agnostic labeling strategy~\cite{van:2017,arnold:2023}: at each sampled point $\gamma \in \Gamma$, divide the parameter space into two sets of points, $\Gamma_{0}(\gamma)$ and $\Gamma_{1}(\gamma)$, each comprised of the $l$ sampled points $\gamma'$ closest to $\gamma$ with $\gamma' \leq \gamma$ and $\gamma'>\gamma$, respectively. Each parameter point $\gamma$ defines a bipartition and, in turn, a classification task. Given a predictive model $\hat{y}: \mathcal{X} \rightarrow [0,1]$ trained to perform this task, its error rate is defined as
\begin{equation}\label{eq:LBC_err}
    p_{\rm err}(\gamma) = \frac{1}{2}\sum_{y \in \{0,1 \}} \frac{1}{|\Gamma_{y}|}\sum_{\gamma' \in \Gamma_{y}}\mathbb{E}_{\bm{x} \sim P(\cdot|\gamma')} \left[{\rm err}(\bm{x},y)\right],
\end{equation}
where ${\rm err}(\bm{x},y)$ is zero if the sample is classified correctly and one otherwise. Intuitively, $p_{\rm err}(\gamma)$ is lowest at a phase boundary where the data is partitioned according to its phase. Thus, in a second approach, critical points can be detected as local maxima in the indicator 
\begin{equation}\label{eq:LBC_ind}
    I_2(\gamma) =  1- 2 p_{\rm err}(\gamma).
\end{equation}

\indent A third approach to detect phase transitions from data is based on parameter estimation~\cite{schaefer:2019,arnold:2023}. At its core lies a predictive model $\hat{\gamma}$ that estimates the parameter $\gamma$ at which a given sample $\bm{x}$ was drawn. Intuitively, the mean predicted value of the tuning parameter
    $\hat{\gamma}(\gamma) = \mathbb{E}_{\bm{x} \sim P(\cdot|\gamma)}\left[ \hat{\gamma}(\bm{x})\right]$
is expected to be most sensitive at phase boundaries. The following indicator captures this susceptibility
\begin{equation}\label{eq:IPBM}
    I_{3}(\gamma) = \frac{\partial \hat{\gamma}(\gamma) / \partial \gamma }{\sigma(\gamma)},
\end{equation}
where $\sigma(\gamma) = \sqrt{\mathbb{E}_{\bm{x} \sim P(\cdot|\gamma)}\left[ \hat{\gamma}(\bm{x})^2\right] - \hat{\gamma}(\gamma)^2}$~\cite{std}.

We have formulated the problem of detecting a phase transition as the computation of an indicator function [Eqs.~\eqref{eq:SL_ind},~\eqref{eq:LBC_ind}, and~\eqref{eq:IPBM}]. This computation involves solving a classification or regression task, i.e., finding a suitable predictive model. This model can be constructed in a data-driven way given a set of samples $\mathcal{D}_{\gamma}$ drawn from $P(\cdot|\gamma)$ for each $\gamma \in \Gamma$. Typically, a parametric approach is chosen in which the predictive model ($\hat{y}$ or $\hat{\gamma}$) is an NN whose parameters $\bm{\theta}$ are optimized in a supervised fashion via the minimization of a loss function $\mathcal{L}(\bm{\theta})$. For approaches 1 and 2 [Eqs.~\eqref{eq:SL_ind} and~\eqref{eq:LBC_ind}] dealing with classification tasks, a typical choice is an unbiased binary cross-entropy loss
\begin{align*}\label{eq:SI_x1}
    \mathcal{L}(\bm{\theta}) = &- \frac{1}{2} \sum_{y \in \{ 0,1\}}\frac{1}{|\mathcal{D}_{y}|} \sum_{\bm{x} \in \mathcal{D}_{y}}  [y{\rm ln}\left(\hat{y}_{\bm{\theta}}(\bm{x})\right)\\
    &+ (1-y){\rm ln}\left(1-\hat{y}_{\bm{\theta}}(\bm{x})\right)]\numberthis. 
\end{align*}
where $\mathcal{D}_{y}$ is composed of all sets of samples $\mathcal{D}_{\gamma}$ with $\gamma \in \Gamma_{y}$. For the regression task in approach 3 [Eq.~\eqref{eq:IPBM}], a mean squared error loss is used
\begin{equation}\label{eq:SI_x2}
    \mathcal{L}(\bm{\theta}) = \frac{1}{|\Gamma|}\sum_{\gamma \in \Gamma}  \frac{1}{|\mathcal{D}_{\gamma}|} \sum_{\bm{x} \in \mathcal{D}_{\gamma}} \left(\gamma - \hat{\gamma}_{\bm{\theta}}(\bm{x})\right)^2.
\end{equation}
Given a predictive model, an estimate of the corresponding indicator can be obtained by replacing expected values with sample means.

\emph{Relating data-driven indicators to the Fisher information}.---In the following, we are going to establish a connection between the three aforementioned indicators of phase transitions [Eqs.~\eqref{eq:SL_ind},~\eqref{eq:LBC_ind}, and~\eqref{eq:IPBM}] and the FI
\begin{equation}\label{eq:FI}
    \mathcal{F}(\gamma) = \mathbb{E}_{\bm{x} \sim P(\cdot|\gamma)}\left[\left(\frac{\partial \log(P(\bm{x}|\gamma))}{\partial \gamma}\right)^2\right], 
\end{equation}
which quantifies the amount of information that the random variable $\bm{x}$ carries about the parameter $\gamma$ characterizing its distribution $P(\cdot|\gamma)$. The intuitive explanation for the existence of such a relationship lies in the fact that all three indicators gauge changes in the underlying probability distributions as a function of the tuning parameter~\cite{arnold:2022}.

\indent The indicator of the first approach can be written as
\begin{equation}\label{eq:SL_bound_1}
    I_1(\gamma) = \left|\mathbb{E}_{\bm{x} \sim P(\cdot|\gamma)}\left[\hat{y}(\bm{x}) \frac{\partial \log(P(\bm{x}|\gamma))}{\partial \gamma}\right] \right|
\end{equation}
using the log-derivative trick. By the Cauchy-Schwarz inequality, the indicator $I_1(\gamma)$ is maximal if and only if $\hat{y}$ is perfectly correlated with the score, i.e.,
\begin{equation}\label{eq:SL_bound_2}
    \frac{\partial \log(P(\bm{x}|\gamma))}{\partial \gamma} = \pm \sqrt{\mathcal{F}(\gamma)} \left( \frac{ \hat{y}(\bm{x}) - \mathbb{E}_{\bm{x}\sim P(\cdot|\gamma)}\left[ \hat{y}(\bm{x}) \right]}{ \sigma(\gamma) }\right), 
\end{equation}
where $\sigma(\gamma)$ is the standard deviation of $\hat{y}$ at $\gamma$ and we have used the fact that the score $\partial \log(P(\cdot|\gamma))/\partial \gamma$ has zero mean and the mean of its square corresponds to the FI [cf. Eq.~\eqref{eq:FI}]. Because samples with, e.g., a negative score, are typically predominantly found in phase 0 compared to phase 1, the correlation of $\hat{y}$ with the score (and thus $I_{1}$ itself) is expected to increase with increasing quality of the predictive model. Based on Eq.~\eqref{eq:SL_bound_2}, we have
\begin{equation}\label{eq:SL_bound}
    I_1(\gamma) \leq \sigma(\gamma) \sqrt{\mathcal{F}(\gamma)} \leq \sqrt{\mathcal{F}(\gamma)},
\end{equation}
where the second inequality follows from the fact that $|\hat{y}(\bm{x})| \leq 1 \; \forall \bm{x} \in \mathcal{X}$ by construction for any valid predictive model. Equation~\eqref{eq:SL_bound} further suggests an improvement of the first approach by modifying its indicator as $I_1(\gamma) \mapsto I_1(\gamma)/\sigma(\gamma)$, as worked out in detail in~\cite{suppl}.

\indent The second approach involves the statistical task of single-shot symmetric binary hypothesis testing: based on a single measurement outcome $\bm{x} \in \mathcal{X}$, determine which of two probability distributions $P$ and $Q$ is more likely to describe the experiment and avoid both false positives and negatives equally. The optimal error probability is given by ${p}_{\rm err}^{\rm opt} = \frac{1}{2}\left(1-{\rm TV}(P,Q)\right)$~\cite{suppl}, where ${\rm TV}$ denotes the total variation distance ${\rm TV}(P,Q) = \frac{1}{2} \sum_{\bm{x}\in \mathcal{X}} |P(\bm{x}) - Q(\bm{x}) |$.
The indicator value corresponding to a given bipartition can thus be upper-bounded as
\begin{equation}\label{eq:LBC_bound}
    I_{\rm 2}(\gamma) \leq 1-2{p}_{\rm err}^{\rm opt}(\gamma) = I_{\rm 2}^{\rm opt}(\gamma) = {\rm TV}(P_{0},P_{1}),
\end{equation}
where the total variation distance is measured between the probability distributions underlying the two partitions $P_{y} = \frac{1}{|\Gamma_{y}|} \sum_{\gamma' \in \Gamma_{y}} P(\cdot|\gamma'), \; y \in \{ 0,1\}$. This bound holds for any valid predictive model $\hat{y}: \mathcal{X} \rightarrow [0,1]$. Note that $p_{\rm err}^{\rm opt}$ is achieved by a predictive model that minimizes the loss in Eq.~\eqref{eq:SI_x1} in the infinite data limit. Thus, approach 2 relies on a variational lower bound of the total variation distance, which improves during training (i.e., as ${p}_{\rm err}$ decreases).

\indent Consider now the behavior of the total variation distance for probability distributions separated by a small distance $\delta \gamma$ in parameter space, i.e., $P(\cdot|\gamma)$ and $P(\cdot|\gamma + \delta \gamma)$,
\begin{equation}
    {\rm TV}\left(P(\cdot|\gamma), P(\cdot|\gamma + \delta \gamma)\right) =  \frac{1}{2} \sum_{\bm{x} \in \mathcal{X}} \left|\frac{\partial P(\bm{x}|\gamma)}{\partial \gamma}\right| \delta \gamma + \mathcal{O}(\delta \gamma^2).
\end{equation}
Using the Cauchy-Schwarz inequality, we have
\begin{equation}
    {\rm TV}\left(P(\cdot|\gamma), P(\cdot|\gamma + \delta \gamma)\right) \leq \frac{1}{2}\sqrt{\mathcal{F}(\gamma)}\delta\gamma + \mathcal{O}(\delta\gamma^2).
\end{equation}
Plugging into Eq.~\eqref{eq:LBC_bound}, we obtain
\begin{equation}
    I_{2}(\gamma) \leq I_{2}^{\rm opt}(\gamma) \leq \frac{1}{2}\sqrt{\mathcal{F}(\gamma)}\delta\gamma + \mathcal{O}(\delta\gamma^2),
\end{equation}
which corresponds to a scenario where $l=1$ and the probability distributions $P_{0}$ and $P_{1}$ corresponding to the two singletons $\Gamma_{0} = \{\gamma \}$ and $\Gamma_{1} = \{ \gamma + \delta \gamma \}$ are separated by a distance $\delta \gamma$. As $\delta\gamma \rightarrow 0$, the indicator of approach 2 (rescaled by $2/\delta \gamma$) serves as a lower bound to the square root of the FI.

The third approach is based on parameter estimation. The Cramér–Rao bound~\cite{cr1,cr2} is a well-known lower bound on the variance of an estimator of a deterministic (fixed, though unknown) parameter $\gamma$:
\begin{equation}
    \sigma^2(\gamma)  \geq \frac{\left(1+\partial b(\gamma)/\partial \gamma\right)^2}{\mathcal{F}(\gamma)} = \frac{(\partial \hat{\gamma}(\gamma)/\partial \gamma)^2}{\mathcal{F}(\gamma)}.
\end{equation}
Here, $b(\gamma) = \mathbb{E}_{\bm{x} \sim P(\cdot|\gamma)}[\hat{\gamma}(\bm{x})- \gamma]$ is the bias of the estimator. Rearranging the equation and plugging in the definition of the indicator of the second approach, we have
\begin{equation}\label{eq:PBM_FI}
    I_{3}(\gamma) = \frac{\partial \hat{\gamma}(\gamma) / \partial \gamma }{\sigma(\gamma)} \leq \sqrt{\mathcal{F}(\gamma)}.
\end{equation}
Interestingly, both $\partial \hat{\gamma}(\gamma) / \partial \gamma$~\cite{schaefer:2019,greplova:2020,arnold:2021,arnold:2022} and $\sigma(\gamma)$~\cite{guo:2023,guo2:2023} have been used separately as indicators of phase transitions, i.e., quantities whose local maxima and minima, respectively, indicate critical points. The connection to the FI established in Eq.~\eqref{eq:PBM_FI} further justifies using their ratio as an indicator of phase transitions. Note that minimum mean-squared error estimation~\cite{alves:2022} is a widely used approach for solving parameter estimation tasks. As such, the bound in Eq.~\eqref{eq:PBM_FI} is generally expected to improve during training using the loss function in Eq.~\eqref{eq:SI_x2}.

The relations between the three data-driven indicators and the FI derived above constitute the central result of our Letter. Next, we investigate this relationship in different physical contexts and provide numerical evidence for the quality of our bounds. As concrete examples, we consider a classical and a quantum model: the two-dimensional Ising model and the one-dimensional transverse-field Ising model.

\emph{Classical equilibrium systems}.---For a system at equilibrium with a large thermal reservoir, $P(\bm{x}| \bm{\gamma}) = e^{-\mathcal{H}(\bm{x},\bm{\gamma})}/Z(\bm{\gamma})$,
where $\bm{x} \in \mathcal{X}$ denotes a configuration of the system, $Z(\bm{\gamma})$ is the partition function, and $\bm{\gamma}=(\gamma_{1},\gamma_{2},\dots,\gamma_{d})$ are tunable parameters, such as the temperature or magnetic field strength. Typically, the dimensionless Hamiltonian $\mathcal{H} = H/k_{\rm B}T$ takes the form $\mathcal{H} = \sum_{i=1}^{d} \gamma_{i} X_{i}(\bm{x})$,
where $X_{i}(\bm{x})$ is a collective variable coupled to the tuning parameter $\gamma_{i}$. The FI $\mathcal{F}_{i}$ associated with the parameter $\gamma_{i}$ can be shown to measure changes in these collective variables, $\mathcal{F}_{i}(\bm{\gamma}) = -\frac{\partial \langle X_{i} \rangle}{\partial \gamma_{i}}$~\cite{prokopenko:2011}. Moreover, since $\frac{\partial A}{\partial \gamma_{i}} = k_{B}T \langle X_{i} \rangle$, where $A$ is the Helmholtz free energy, we have $\mathcal{F}_{i}(\bm{\gamma}) = - \beta \frac{\partial^2 A}{\partial \gamma_{i}^2}$, where $\beta = 1/k_{\rm B}T$. Because the FI is related to second derivatives of the free energy, it is sensitive to first- and second-order divergences. In the case of thermal transitions where the tuning parameter is a function of $T$, for example, the FI is proportional to the heat capacity $\mathcal{F} \propto C$.

\begin{figure}[tbh!]
	\centering
		\includegraphics[width=0.99\linewidth]{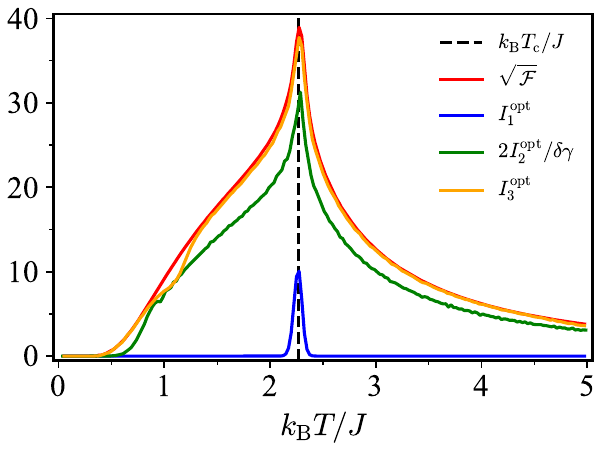}
		\caption{Results of the three data-driven approaches for detecting phase transitions [Eqs.~\eqref{eq:SL_ind},~\eqref{eq:LBC_ind}, and~\eqref{eq:IPBM}] applied to the square-lattice ferromagnetic Ising model ($L = 60$). Here, we consider (approximate) Bayes-optimal predictive models. The critical point $k_{\rm B}T_{\rm c}/J = 2/ \ln(1+\sqrt{2})$ is highlighted by a black-dashed line. All shown quantities are lower bounds to the square root of the system's FI. The set $\Gamma$ is composed of a uniform grid with 200 points and grid spacing $\delta \gamma = 0.025$. Each dataset $D_{\gamma}$ consists of $10^6$ spin configurations. For approach 1 [Eq.~\eqref{eq:SL_ind}], we choose $\Gamma_{0} = \{0.025,\dots,2.25\}$ and $\Gamma_{1} = \{2.275,\dots,5\}$, i.e., we choose the two regions in parameter space to coincide with the two phases. For approach 2, we choose $l=1$.}
		\label{fig:1}
\end{figure}

\indent As a concrete example, we consider the thermal phase transition in the $L \times L$ square-lattice classical Ising model described by the Hamiltonian $H(\bm{\sigma}) = - J \sum_{\langle ij\rangle} \sigma_{i}\sigma_{j}$, where the sum runs over all nearest-neighboring sites (with periodic boundary conditions), $J$ is the interaction strength $(J>0)$, and $\sigma_{i} \in \{+1,-1 \}$. Here, we choose $\gamma = k_{\rm B}T/J$ as a tuning parameter. The energy is the relevant collective variable $X(\bm{\sigma}) = H(\bm{\sigma})/J$ and the system's FI corresponds to $CJ^2/k_{\rm B}^3T^2$ where $C$ is the heat capacity. We draw spin configurations from $\{ P(\cdot|\gamma)\}_{\gamma \in \Gamma}$ via Markov chain Monte Carlo~\cite{suppl}. Based on the resulting datasets $\{ \mathcal{D}_{\gamma}\}_{\gamma \in \Gamma}$, we construct (approximate) \textit{Bayes-optimal} predictive models using nonparametric generative models obtained via histogram binning of the sufficient statistic $X$~\cite{arnold:2022,arnold:2023}. These models constitute global minima of the relevant loss functions [Eqs.~\eqref{eq:SI_x1} and~\eqref{eq:SI_x2}]. Figure~\ref{fig:1} shows the indicators $I^{\rm opt}$ corresponding to these Bayes-optimal predictive models for the Ising model. They are good underapproximators of the square root of the system's FI showing similar functional behavior. In particular, their peak positions are in agreement. Note that the indicator of approach 1 (including its peak position) depends heavily on the choice of $\Gamma_{0}$ and $\Gamma_{1}$, i.e., on prior knowledge of the location of the phase transition~\cite{suppl}.

\emph{Quantum systems}.---In quantum physics, measurements are described by a positive operator-valued measure (POVM). The probability of obtaining the measurement outcome $\bm{x} \in \mathcal{X}$ associated with the POVM element $\Pi_{\bm{x}}$ is given by $P(\bm{x}|\gamma) = {\rm tr}\left(\Pi_{\bm{x}} \rho(\gamma)\right)$. The \emph{quantum} FI corresponds to the classical FI maximized over all possible measurements, $\mathcal{F}^{Q}(\gamma) = \mathcal{F}^{Q}(\rho(\gamma)) = \max_{\{\Pi_{\bm{x}}\}_{\bm{x}}} \mathcal{F}\left(P(\cdot|\gamma)\right)$~\cite{braunstein:1994}. Moreover, expanding the fidelity $ F(\rho,\sigma) = {\rm tr}(\sqrt
{\sqrt{\sigma}\rho\sqrt{\sigma}})$ between infinitesimally close states~\cite{jozsa:1994}, we have $F(\rho(\gamma),\rho(\gamma+\delta \gamma)) = 1 - \delta \gamma^2 \mathcal{F}^{Q}(\rho(\gamma))/8 + \mathcal{O}(\delta \gamma^3) = 1 - \delta \gamma^2 \chi_{\mathcal{F}}(\rho(\gamma))/2 + \mathcal{O}(\delta \gamma^3) $, where $\chi_{\mathcal{F}} = \mathcal{F}^{Q}(\rho(\gamma))/4$ is the fidelity susceptibility~\cite{gu:2010,liu:2014}.

As an example, we consider the transverse-field Ising model~\cite{sachdev:2011} on a (periodic) one-dimensional chain of length $L$ whose Hamiltonian is given by $H = -J \sum_{\langle ij\rangle} \sigma_{i}^{z}\sigma_{j}^{z} - h \sum_{i} \sigma_{i}^{x}$,
where $J>0$ is the nearest-neighbor interaction strength, $h$ is the external field strength, and $\{\sigma_{i}^{x}, \sigma_{i}^{y}, \sigma_{i}^{z}\}$ are the Pauli operators acting on the spin at site $i$. This model undergoes a quantum phase transition at zero temperature from a ferromagnetically ordered phase at $\gamma =  h/J<1$ to a disordered phase. We perform exact diagonalization and consider projective measurements in the $x$-basis. The optimal indicators of all three approaches yield non-trivial bounds on the square root of the system's classical and quantum FI and their peak positions agree, see Fig.~\ref{fig:2}. Applying classical data-driven methods to quantum systems requires the choice of a POVM. Many previous works have successfully detected phase transitions using simple projective measurements in a single basis~\cite{greplova:2020,miles:2021,bohrdt:2021,miles2:2021,maskara:2022}. Our findings highlight that a good choice of measurement is one that results in a high classical FI, i.e., one for which the latter is close to the quantum FI.\\

\begin{figure}[bth!]
	\centering
		\includegraphics[width=0.99\linewidth]{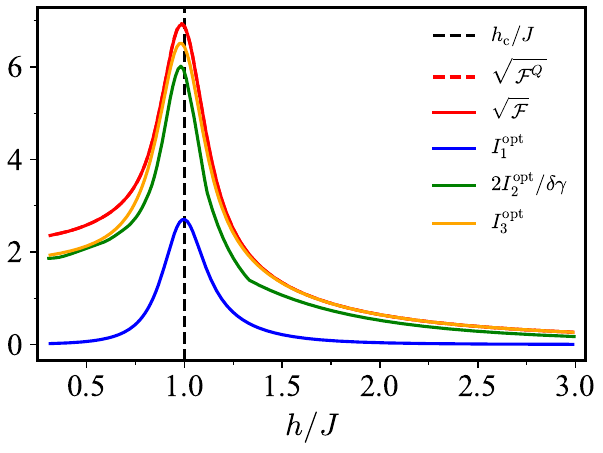}
		\caption{Results of the three data-driven approaches for detecting phase transitions [Eqs.~\eqref{eq:SL_ind},~\eqref{eq:LBC_ind}, and~\eqref{eq:IPBM}] applied to the one-dimensional transverse-field Ising model ($L = 20$). Here, we consider Bayes-optimal predictive models constructed from the exact probability distributions underlying the measurement statistics~\cite{arnold:2022,arnold:2023}, which globally minimize the corresponding loss functions. The critical point $h_{c}/J=1$ is highlighted by a black-dashed line~\cite{sachdev:2011}. All shown quantities are lower bounds to the square root of the system's classical and quantum FI (in the figure these two curves overlap). The set $\Gamma$ is composed of a uniform grid with 201 points and grid spacing $\delta \gamma = 0.0135$. For approach 1 [Eq.~\eqref{eq:SL_ind}], we choose $\Gamma_{0} = \{0.3\}$ and $\Gamma_{1} = \{3\}$, i.e., we choose the two sets to be composed of single points at the edges of the sampled region in parameter space. For approach 2, we choose $l=1$.}
		\label{fig:2}
\end{figure}

\emph{Conclusion}.---We have unveiled a fundamental connection between the (so far disparate) information-theoretic and machine-learning paradigms of studying critical phenomena: The indicators of phase transitions of several popular ML methods approximate the square root of the system's FI from below. We numerically demonstrated the quality of these underapproximations for phase transitions in classical equilibrium systems and quantum ground states. Our result sheds light on the fundamental working principle and limitations of the ML methods considered in this work:

We find that they can be viewed as data-driven approaches for constructing approximations of the FI that remain operationally useful for identifying phase transitions. The exact FI is difficult to access if only samples of the system are available, as computing the FI generally requires knowledge of the probabilistic model underlying the measurement statistics of the system [see Eq.~\eqref{eq:FI}].

As an example of limitations, a variety of previous numerical studies~\cite{beach:2018,suchsland:2018,arnold:2022,guo:2023} have found that NN-based methods struggle to detect thermal transitions of the BKT-type in classical equilibrium systems. In particular, the NN-based indicators have been observed to show a peak at the same position as the heat capacity, away from the critical point. The relation between NN-based indicators and the FI we have uncovered offers a natural explanation for this observation, given that the FI reduces to the heat capacity in such a case while the BKT transition is of infinite order. More generally, the FI is related to second derivatives of the free energy and is thus only sensitive to first- and second-order divergences. Our work suggests that the same holds for the ML methods we considered.

\emph{Outlook}.---Interestingly, data-driven schemes for estimating the FI based on approximating different statistical divergences (akin to the variational lower bound of the total variation distance utilized in approach 2) have recently been proposed~\cite{berisha:2014,duy:2022}. Similar ML methods based on variational representations have also been utilized to estimate other classical~\cite{nguyen:2007,belghazi:2018} as well as quantum information-theoretic quantities~\cite{cerezo:2020,tan:2021,beckey:2022,shin:2023,goldfeld:2023}. In light of our results, such approaches may give rise to a whole new set of methods to detect phase transitions from data that have so far gone unnoticed. 

Here, we have focused on three particular ML approaches for detecting phase transitions. It will be interesting to see what other methods are related to the FI. For example, the generative adversarial network fidelity defined in Ref.~\cite{singh:2021} can also be shown to approximate the FI (see Ref.~\cite{suppl} for details). We anticipate that information theory will be instrumental in understanding, categorizing, and improving the growing number of ML methods for detecting phase transitions, with our results forming the basis for such efforts.

\indent We thank Alexander Gresch, Lennart Bittel, Martin Kliesch, and Christoph Bruder for stimulating discussions. J.A. acknowledges financial support from the Swiss National Science Foundation individual grant (grant no. 200020 200481). Computation time at sciCORE (scicore.unibas.ch) scientific computing center at the University of Basel is gratefully acknowledged. This material is based upon work supported by the National Science Foundation under grant no. OAC-1835443, grant no. OAC-2103804, and grant no. DMS-2325184.

\let\oldaddcontentsline\addcontentsline
\renewcommand{\addcontentsline}[3]{}
\bibliography{refs.bib}
\let\addcontentsline\oldaddcontentsline

\newpage
\pagebreak

\setcounter{equation}{0}
\setcounter{figure}{0}
\setcounter{table}{0}
\setcounter{page}{1}
\makeatletter
\renewcommand{\thesection}{S\arabic{section}}
\renewcommand{\theequation}{S\arabic{equation}}
\renewcommand{\thefigure}{S\arabic{figure}}

\onecolumngrid

\begin{center}
\textbf{\large Supplemental Material for ``Machine learning phase transitions: Connections to the Fisher information''}
\end{center}

\tableofcontents

\section{Statistical background}
In this section, we will provide a short overview of the relevant statistical background underlying the findings of our work, touching upon concepts such as $f$-divergences, information geometry, hypothesis testing, and parameter estimation. Readers familiar with these topics may skip this section. For a more extended overview, see, for example, Refs.~\cite{casella:2002,bickel:2015,jarzyna:2020}.\\

The space composed of all valid probability distributions on a given probability space is referred to as a \emph{statistical manifold} $\mathcal{M}$. The field of information geometry is concerned with the geometry of this manifold. It can give useful insights for dealing with statistical inference tasks such as parameter estimation or hypothesis testing: many machine learning (ML) methods for detecting phase transitions make use of such tasks, including the three methods discussed in the main text [see Eqs. (1), (3), and (4) in the main text].\\

\indent Let us start by reviewing the notion of \emph{statistical distances}. They measure the distance between statistical objects, such as probability distributions. A statistical distance between two elements of the statistical manifold is some non-negative function $D: \mathcal{M} \times \mathcal{M} \rightarrow \mathbb{R}_{\geq 0}$. A statistical distance is a proper distance or \emph{metric} if, in addition, it satisfies $i)$ \emph{symmetry}: $D[p,q] = D[q,p]$, $ii)$ \emph{identity of indiscernibles}: $D[p,q] = 0 \iff p=q$, and $iii)$  \emph{triangle inequality}: $D[p,r] + D[r,q] \geq D[p,q]$, for any valid probability distributions $p,q,r \in \mathcal{M}$ defined over the state space $\mathcal{X}$. For simplicity, in what follows we assume the state space $\mathcal{X}$ to be discrete and countable. It turns out that many statistical distances of interest do not satisfy all these criteria. An important class of statistical distances are the so-called $f$-divergences~\cite{liese:2006}. 
\begin{definition}[$f$-divergence]
    Given a convex function $f: \mathbb{R}_{\geq 0} \rightarrow \mathbb{R}$ with $f(1)=0$, the corresponding $f$-divergence is a statistical distance defined as
\begin{equation}
    D_{f}[p,q] = \sum_{\bm{x} \in \mathcal{X}} q(\bm{x}) f\left(\frac{p(\bm{x})}{q(\bm{x})}\right).
\end{equation}
\end{definition}
In general, an $f$-divergence does not constitute a proper metric.\footnote{The non-negativity of $f$-divergences follows from their convexity via Jensen's inequality. Moreover, if $f(x)$ is strictly convex at $x=1$, the corresponding $f$-divergence can be shown to satisfy the identity of indiscernible, which justifies referring to $D_{f}$ as a divergence.}\\

\indent The \emph{total variation} (TV) distance is an $f$-divergence which will become particularly useful for us later on. It is defined as
\begin{equation}\label{eq:TV_distance}
    {\rm TV}[p,q] = \frac{1}{2} \sum_{\bm{x}\in \mathcal{X}} |p(\bm{x}) - q(\bm{x})|.
\end{equation}
The TV distance is the $f$-divergence with $f(x) = \frac{1}{2}|1-x|$. In contrast to other $f$-divergences, the function $f$ of the TV distance is not differentiable at $1$. Other important examples of $f$-divergences include the \emph{Kullback-Leibler} (KL) divergence
\begin{equation}
    {\rm KL}[p,q] = \sum_{\bm{x} \in \mathcal{X}} p(\bm{x})\ln\left( \frac{p(\bm{x})}{q(\bm{x})}\right),
\end{equation}
with $f(x) = x\ln(x)$ and the \emph{Jensen-Shannon} (JS) divergence
\begin{equation}\label{eq:JSD}
    {\rm JS}[p,q] = \frac{1}{2}{\rm KL}\left[p,\frac{p+q}{2}\right] + \frac{1}{2}{\rm KL}\left[q,\frac{p+q}{2}\right],
\end{equation}
with $f(x) = \frac{1}{2}\left(x\ln(\frac{2x}{1+x}) +  \ln(\frac{2}{1+x})\right)$. Note that the generating function $f$ of a given $f$-divergence is not uniquely defined, but only up to an affine term. That is, $D_{f} = D_{g}$ if $f(x) = g(x) + c(x-1)$ for some $c \in \mathbb{R}$.\\

While $f$-divergences are not proper metrics, they satisfy other crucial properties. For example, a good statistical distance should capture the information loss associated with data processing. As such, it should fulfill the so-called data-processing inequality.

\begin{proposition}[Data-processing inequality]
    Consider a mapping from $\mathcal{X}$ to an alternative space $\mathcal{Y}$, $S:\mathcal{X} \rightarrow \mathcal{Y}$, such that $p(\bm{y}) = \sum_{\bm{x} \in \mathcal{X}} W(\bm{y}|\bm{x})p(\bm{x})$ with $W$ being a left-stochastic transition matrix, i.e., a matrix with non-negative entries and columns summing up to one. Then, $D_{f}[p(\bm{x}),q(\bm{x})] \geq D_{f}[p(\bm{y}),q(\bm{y})]$ for any $f$-divergence $D_{f}$.
\end{proposition}

\begin{proof}
\begin{align*}
    D_{f}[p(\bm{x}),q(\bm{x})] &= \sum_{\bm{x} \in \mathcal{X}} q(\bm{x})f\left(\frac{p(\bm{x})}{q(\bm{x})}\right) = \sum_{\bm{x}\in \mathcal{X},\bm{y}\in \mathcal{Y}} q(\bm{x})W(\bm{y}|\bm{x}) f\left(\frac{p(\bm{x})W(\bm{y}|\bm{x})}{q(\bm{x})W(\bm{y}|\bm{x})}\right)
    = \sum_{\bm{x},\bm{y}} q(\bm{x},\bm{y}) f\left(\frac{p(\bm{x},\bm{y})}{q(\bm{x},\bm{y})}\right)\\
    &= D_{f}[p(\bm{x},\bm{y}),q(\bm{x},\bm{y})] = \sum_{\bm{y}} q(\bm{y}) \sum_{\bm{x}} q(\bm{x}|\bm{y}) f\left(\frac{p(\bm{y})p(\bm{x}|\bm{y})}{q(\bm{y})q(\bm{x}|\bm{y})}\right)
    \leq \sum_{\bm{y}} q(\bm{y}) f\left(\sum_{\bm{x}} q(\bm{x}|\bm{y}) \frac{p(\bm{y})p(\bm{x}|\bm{y})}{q(\bm{y})q(\bm{x}|\bm{y})}\right)
\end{align*}
where we have used Jensen's inequality in the last step. Finally noting that
\begin{align*}
    \sum_{\bm{y}} q(\bm{y}) f\left(\sum_{\bm{x}} q(\bm{x}|\bm{y}) \frac{p(\bm{y})p(\bm{x}|\bm{y})}{q(\bm{y})q(\bm{x}|\bm{y})}\right) =  \sum_{\bm{y}} q(\bm{y}) f\left(\sum_{\bm{x}} p(\bm{x}|\bm{y}) \frac{p(\bm{y})}{q(\bm{y})}\right) = \sum_{\bm{y}} q(\bm{y}) f\left( \frac{p(\bm{y})}{q(\bm{y})}\right) = D_{f}[p(\bm{y}),q(\bm{y})]
\end{align*}
completes the proof. 
\end{proof}

The intuition is that processing $\bm{x}$ (via a local physical operation described by a Markov process) can only make it more difficult to distinguish two distributions. Note that statistical distances that satisfy the data processing inequality are also called \emph{monotonic} (under stochastic maps).\\

\indent A good statistical distance should also be invariant under mappings between sample spaces that preserve all ``relevant information'' about $\bm{x}$. To this end, let us endow the statistical manifold $\mathcal{M}$ with a coordinate system by parametrizing all probability distributions in the manifold $p \mapsto p_{\bm{\gamma}}, \bm{\gamma} \in \mathbb{R}^{d}$, where $d = {\mathrm{dim}}\; \mathcal{M}$. In this case, a statistic that encodes all relevant information about the value of the parameter $\bm{\gamma}$ is called sufficient. The statistic is sufficient in the sense that there does not exist any other statistic that could be calculated from the samples $\bm{x}$ that would provide additional information regarding the value of the parameter.
\begin{definition}[Sufficient statistic]
    Given a mapping $S: \mathcal{X} \rightarrow \mathcal{Y}$, the statistic $S(\bm{x})$ is sufficient for $\bm{\gamma}$ if and only if $\forall \bm{\gamma}, \bm{y}$ we have that $p_{\bm{\gamma}}(\bm{x}|\bm{y} = S(\bm{x}))$ is independent of $\bm{\gamma}$.
\end{definition}

Having specified what is meant by relevant information, we can verify that $f$-divergences are indeed invariant under mappings between sample spaces that leave this information intact. In particular, $f$-divergences can be shown to be invariant under mappings $S: \mathcal{X} \rightarrow \mathcal{Y}$ where the statistic $S(\bm{x})$ is sufficient for $\bm{\gamma}$.

\begin{proposition}[Invariance under sufficient statistic]
    Consider a mapping $S: \mathcal{X} \rightarrow \mathcal{Y}$ where the statistic $S(\bm{x})$ is sufficient for $\bm{\gamma}$. Then, $D_{f}[p_{\bm{\gamma}_{1}}(\bm{x}),p_{\bm{\gamma}_{2}}(\bm{x})] = D_{f}[p_{\bm{\gamma}_{1}}(\bm{y}),p_{\bm{\gamma}_{2}}(\bm{y})]$ for any choice of $\bm{\gamma}_{1}$, $\bm{\gamma}_{2}$ and $f$-divergence.
\end{proposition}
\begin{proof}
    Any $f$-divergence satisfies the data-processing inequality. During data processing, equality is achieved if $\frac{p(\bm{x}|\bm{y})}{q(\bm{x}|\bm{y})} = 1$. Choosing $p = p_{\bm{\gamma}_{1}}$ and $q = p_{\bm{\gamma}_{2}}$, this is satisfied given that $S(\bm{x})$ is a sufficient statistic.
\end{proof}

The Fisher-Neyman factorization theorem provides another convenient characterization of a sufficient statistic.

\begin{theorem}[Fisher-Neyman factorization theorem]
    Given a mapping $S: \mathcal{X} \rightarrow \mathcal{Y}$, the statistic $S(\bm{x})$ is sufficient for $\bm{\gamma}$ if and only if non-negative functions $h$ and $g$ can be found such that $p_{\bm{\gamma}}(\bm{x}) = h(\bm{x})g_{\bm{\gamma}}(S(\bm{x})) \ \forall \bm{\gamma}, \bm{x}$.
\end{theorem}

\begin{proof}
    \noindent ($\implies$): If $\bm{y} = S(\bm{x})$, we have $p_{\bm{\gamma}}(\bm{x},\bm{y}) = p_{\bm{\gamma}}(\bm{x}) p_{\bm{\gamma}}(\bm{y}|\bm{x}) = p_{\bm{\gamma}}(\bm{x})$ given that $p_{\bm{\gamma}}(\bm{y}|\bm{x}) = 1$. Thus, $p_{\bm{\gamma}}(\bm{x})=p_{\bm{\gamma}}(\bm{x},\bm{y})=p_{\bm{\gamma}}(\bm{x}|\bm{y})p_{\bm{\gamma}}(\bm{y}) = p(\bm{x}|\bm{y}) p_{\bm{\gamma}}(\bm{y}),$ where the last equality follows by invoking that $\bm{y} = S(\bm{x})$ is a sufficient statistic. Therefore, with $h(\bm{x}) \equiv p(\bm{x}|\bm{y})$ and $g_{\bm{\gamma}}(\bm{y}) \equiv p_{\bm{\gamma}}(\bm{y})$, we have $p_{\bm{\gamma}}(\bm{x})= h(\bm{x}) g_{\bm{\gamma}}(S(\bm{x}))$.\\
    
    \noindent ($\impliedby$): $p_{\bm{\gamma}}(\bm{y})=\sum_{\bm{x} \in \mathcal{X};\; S(\bm{x})=\bm{y}}  p_{\bm{\gamma}}(\bm{x},\bm{y}) = \sum_{\bm{x} \in \mathcal{X};\; S(\bm{x})=\bm{y}}  p_{\bm{\gamma}}(\bm{x})$. Inserting $p_{\bm{\gamma}}(\bm{x})=h(\bm{x})g_{\bm{\gamma}}(\bm{y})$, we have $p_{\bm{\gamma}}(\bm{y})=\sum_{\bm{x} \in \mathcal{X};\; S(\bm{x})=\bm{y}}  h(\bm{x})g_{\bm{\gamma}}(\bm{y}) = \left(\sum_{\bm{x} \in \mathcal{X};\; S(\bm{x})=\bm{y}}  h(\bm{x})\right)g_{\bm{\gamma}}(\bm{y})$. The quantity $p_{\bm{\gamma}}(\bm{x}|\bm{y}) = \frac{p_{\bm{\gamma}}(\bm{x},\bm{y})}{p_{\bm{\gamma}}(\bm{y})} =\frac{p_{\bm{\gamma}}(\bm{x})}{p_{\bm{\gamma}}(\bm{y})}  $ is thus independent of $\bm{\gamma}$, $p_{\bm{\gamma}}(\bm{x}|\bm{y})  =  \frac{h(\bm{x})g_{\bm{\gamma}}(\bm{y})}{ \left(\sum_{\bm{x} \in \mathcal{X};\; S(\bm{x})=\bm{y}}  h(\bm{x})\right)g_{\bm{\gamma}}(\bm{y}) } = \frac{h(\bm{x})}{ \left(\sum_{\bm{x} \in \mathcal{X};\; S(\bm{x})=\bm{y}}  h(\bm{x})\right)} $. This proves that $S(\bm{x})$ is a sufficient statistic.
\end{proof}

Crucially, the factorized form guaranteed by the Fisher-Neyman factorization theorem, i.e., the fact that the dependence of $\bm{x}$ on $\bm{\gamma}$ only enters through the sufficient statistic $S(\bm{x})$, implies that Bayes-optimal estimates of parameters as well as strategies in hypothesis testing only depend on the sufficient statistic. That is, the optimal indicators of phase transitions of the three methods discussed in the main text can be computed solely from the sufficient statistic as opposed to measurements of the full state space~\cite{arnold:2023}. Similarly, the Fisher information with respect to $\bm{\gamma}$ can be shown to be invariant under such a mapping.
\subsection{Information geometry}\label{sec:inf_theory}
Considering any statistical distance $D[p,q]$ smooth in $p,q \in \mathcal{M}$. Then, we have
\begin{equation}\label{eq:2nd_order}
    D[p_{\bm{\gamma}},p_{\bm{\gamma}+\bm{\delta \gamma}}] = \frac{1}{2}\bm{\delta \gamma}^T H_{D}(\bm{\gamma})\bm{\delta \gamma} + \mathcal{O}(\delta \gamma^3),
\end{equation}
where we use the fact that $D[p,p]=0$ and the first-order term vanishes because $D[p,p]=0$ is a minimum (any statistical distance is non-negative). Here, $H_{D}(\bm{\gamma})=H_{D}(p_{\bm{\gamma}})$ is the Hessian matrix with entries
\begin{equation}
    [H_{D}(\bm{\gamma})]_{i,j} = \left.\frac{\partial^2}{\partial \phi_{i} \partial \phi_{j}} D[p_{\bm{\gamma}},p_{\bm{\phi}}]\right|_{\bm{\phi}=\bm{\gamma}}.
\end{equation}
The components of the Hessian matrix measure how susceptible the probability distribution $p_{\bm{\gamma}}$ is to small changes in the underlying coordinates $\bm{\gamma}$, where the resulting deviations are measured by the statistical distance $D$. One can show that any Hessian induced by a monotonic statistical distance (such as $f$-divergences) must also be monotonic. 
\begin{proposition}[Monotonicity of Hessian]
    Consider a stochastic map $\mathcal{S}: \mathcal{M} \rightarrow \mathcal{M}$ such that $p' = \mathcal{S}p$, where $p'(\bm{y}) = \sum_{\bm{x} \in \mathcal{X}} W(\bm{y}|\bm{x}) p(\bm{x})$ with $W$ being a left-stochastic transition matrix. The Hessian of any sufficiently smooth statistical distance $D$ monotonic under such maps must also be monotonic.
\end{proposition}
\begin{proof}
Because $D$ is monotonic, it satisfies the data-processing inequality. Thus, we have $D[p_{\bm{\gamma}}',p_{\bm{\gamma} + \bm{\delta \gamma}}'] \leq D[p_{\bm{\gamma}},p_{\bm{\gamma} + \bm{\delta \gamma}}]$. Expanding the statistical distance to second order according to Eq.~\eqref{eq:2nd_order}, we have that $H_{D}(p_{\bm{\gamma}}) \geq H_{D}(\mathcal{S}p_{\bm{\gamma}})$ up to $\mathcal{O}(\bm{\delta \gamma}^3)$. Letting  $\delta \gamma \rightarrow 0$ concludes the proof.
\end{proof}
Moreover, if $D$ is an $f$-divergence with $f'(1)= 0$, the Hessian can be shown to be proportional to the (classical) Fisher information matrix $\mathcal{F}(\bm{\gamma}) = \mathcal{F}(p_{\bm{\gamma}})$, where
\begin{equation}
     \mathcal{F}_{i,j}(\bm{\gamma}) = \mathcal{F}_{i,j}(p_{\bm{\gamma}}) = \mathbb{E}_{\bm{x} \sim p_{\bm{\gamma}}}\left[\left(\frac{\partial \log(p_{\bm{\gamma}})}{\partial \gamma_{i}}\right) \left(\frac{\partial \log(p_{\bm{\gamma}})}{\partial \gamma_{j}}\right)\right].
\end{equation}
\begin{proposition}[Relation between Hessian and Fisher information matrix]
    The Hessian matrix of any $f$-divergence $D_{f}$ with $f$ being twice-differentiable is given by $H_{D_{f}}(\bm{\gamma}) = f''(1) \mathcal{F}(\bm{\gamma})$.
\end{proposition}
\begin{proof}
$[H_{D_{f}}(\bm{\gamma})]_{i,j} = \left.\frac{\partial^2}{\partial \phi_{i} \partial \phi_{j}} D_{f}[p_{\bm{\gamma}},p_{\bm{\phi}}]\right|_{\bm{\phi}=\bm{\gamma}} = \sum_{\bm{x}\in\mathcal{X}} \frac{1}{p_{\bm{\gamma}}(\bm{x})} f''(1) \frac{\partial p_{\bm{\gamma}}(\bm{x}) }{\partial \gamma_{i}} \frac{\partial p_{\bm{\gamma}}(\bm{x}) }{\partial \gamma_{j}} = f''(1)\mathcal{F}_{i,j}(\bm{\gamma})$
where we used the fact that $f(1) = f'(1)=0$. If $f'(1) \neq 0$, we may use our freedom in the choice of generating function to ensure otherwise. That is, we replace $f \mapsto g$, where $g(x) = f(x) - f'(1)(x-1)$ retaining $D_{f} = D_{g}$.
\end{proof}
Note that the scalar version of the Fisher information matrix is referred to as Fisher information.\\

Combining the above findings, namely that $f$-divergences are monotonic, the Hessian corresponding to any monotonic statistical distance must also be monotonic, and the Hessian of any $f$-divergence is proportional to the Fisher information matrix, we have that the Fisher information matrix is also monotonic, i.e.,
\begin{equation}
    \mathcal{F}(\mathcal{S}p_{\bm{\gamma}}) \leq \mathcal{F}(p_{\bm{\gamma}}).
\end{equation}
This also implies that the Fisher information remains invariant under the mapping $S: \mathcal{X} \rightarrow \mathcal{Y}$ if $\bm{y} = S(\bm{x})$ is a sufficient statistic.\\

Chentsov's theorem~\cite{chentsov:1978,fujiwara:2022} extends the above argument to all monotonic metrics. It states that all Riemannian metrics defined on a given statistical manifold that are monotonic correspond to the Fisher metric (i.e., the Fisher information matrix) up to a multiplicative constant. Note that $H_{D_{f}}$ can be interpreted as a Riemannian metric on the statistical manifold $\mathcal{M}$. Thus, Chentov's theorem covers the above discussion as a special case. Monotonicity is an important property for any sensible statistical distance. As such, Chentsov's theorem effectively singles out the Fisher information matrix as the only natural metric on the statistical manifold, justifying why many properties of statistical models should be describable in terms of the Fisher information matrix~\cite{dowty:2018}. In the context of detecting phase transitions from data, the above discussion provides some intuition on why many approaches that rely on measuring changes in the underlying probability distributions are ultimately related to the Fisher information. In particular, while one may \textit{a priori} choose distinct $f$-divergences for gauging such changes, as long as the distributions are sufficiently close in parameter space, any such choice reduces to the Fisher information.

\subsection{Hypothesis testing}\label{sec:HT}
Next, we consider the statistical inference task of hypothesis testing. Given the outcomes $\{\bm{x}_{1},\bm{x}_{2},\dots,\bm{x}_{n} \}$ of $n$ independent rounds of an experiment, one needs to decide which distribution from a set of $m$ possible choices $\{p_{1},p_{2},\dots,p_{m}\}$ is most likely to describe the experiment. In the following, we are concerned with binary hypothesis testing, i.e., with distinguishing between two distinct distributions ($m=2$) given a single measurement outcome ($n=1$). Moreover, we treat false positives and false negatives equally. Thus, the task corresponds to \emph{single-shot symmetric binary hypothesis testing} in which one is interested in minimizing the average error probability 
\begin{equation}\label{eq:HT_error}
    p_{\rm err} = \frac{1}{2}P(q|p) + \frac{1}{2}P(p|q), 
\end{equation}
where $P(q|p)$ corresponds to the probability of selecting the probability distribution $q$ while the data actually came from $p$ [and vice versa for $P(p|q)$]. The goal is to find a decision function $\hat{y}: \mathcal{X} \rightarrow \{0,1 \}$, where $\hat{y}=0$ corresponds to the conclusion that $p$ is the correct underlying distribution and $\hat{y}(\bm{x})=1$ corresponds to the conclusion that $q$ is the correct underlying distribution. The optimal inference strategy can be found by minimizing the average error probability $p_{\rm err}$ in Eq.~\eqref{eq:HT_error}. In the scenario we consider, the optimal strategy, also called \emph{Neyman-Pearson} strategy, corresponds to guessing $p$ if $p(\bm{x}) \geq q(\bm{x})$ and guessing $q$ otherwise. Under this strategy, the error probabilities can be expressed as
\begin{equation}\label{eq:err1}
    P(p|q) = \sum_{\substack{\bm{x}\in \mathcal{X}\\
                   p(\bm{x}) \geq q(\bm{x})}}
                q(\bm{x}) \qquad {\rm and} \qquad P(q|p) = \sum_{\substack{\bm{x}\in \mathcal{X}\\
                   q(\bm{x}) > p(\bm{x})}}
                p(\bm{x}).
\end{equation}
\begin{proposition}[Error corresponding to Neyman-Pearson strategy]
    When making predictions according to the Neyman-Pearson strategy, the (optimal) average error probability is equal to \begin{equation}\label{eq:opt_err}
    p_{\rm err}^{\rm opt} = \frac{1}{2}\left(1-\frac{1}{2}\sum_{\bm{x} \in \mathcal{X}} |p(\bm{x}) - q(\bm{x})| \right).
\end{equation}
\end{proposition}
    
\begin{proof}
\begin{align*}
    p_{\rm err}^{\rm opt} = \frac{1}{2}\left(1-\frac{1}{2}\sum_{\bm{x}\in \mathcal{X}} |p(\bm{x})-q(\bm{x})|\right) &= \frac{1}{2}-\frac{1}{4}\sum_{\substack{\bm{x}\in \mathcal{X}\\
                   p(\bm{x}) \geq q(\bm{x})}} \left(p(\bm{x})-q(\bm{x})\right) -\frac{1}{4}\sum_{\substack{\bm{x}\in \mathcal{X}\\
                   q(\bm{x}) > p(\bm{x})}} \left(q(\bm{x})-p(\bm{x})\right)\\
    &= \frac{1}{2}-\frac{1}{4}\left[\sum_{\substack{\bm{x}\in \mathcal{X}\\
                   p(\bm{x}) \geq q(\bm{x})}} p(\bm{x}) +\sum_{\substack{\bm{x}\in \mathcal{X}\\
                   q(\bm{x}) > p(\bm{x})}} q(\bm{x}) - \sum_{\substack{\bm{x}\in \mathcal{X}\\
                   p(\bm{x}) \geq q(\bm{x})}} q(\bm{x})   - \sum_{\substack{\bm{x}\in \mathcal{X}\\
                   q(\bm{x}) > p(\bm{x})}} p(\bm{x})\right].
\end{align*}
Using $1 = \sum_{\bm{x}\in \mathcal{X}} p(\bm{x}) = \sum_{\substack{\bm{x}\in \mathcal{X}\\
                   p(\bm{x}) \geq q(\bm{x})}} p(\bm{x}) + \sum_{\substack{\bm{x}\in \mathcal{X}\\
                   q(\bm{x}) > p(\bm{x})}} p(\bm{x})$ (and similarly for $q$) to rewrite the first two terms, we have
                   \begin{equation*}
                       p_{\rm err}^{\rm opt} = \frac{1}{2}\sum_{\substack{\bm{x}\in \mathcal{X}\\
                   p(\bm{x}) \geq q(\bm{x})}} q(\bm{x}) +\frac{1}{2}\sum_{\substack{\bm{x}\in \mathcal{X}\\
                   q(\bm{x}) > p(\bm{x})}}p(\bm{x}),
                   \end{equation*}
corresponding to the error rate under the optimal strategy obtained by plugging in Eq.~\eqref{eq:err1} into Eq.~\eqref{eq:HT_error}.
\end{proof}
Using the definition of the TV distance in Eq.~\eqref{eq:TV_distance}, the minimal average error probability [Eq.~\eqref{eq:opt_err}] can be expressed as 
\begin{equation}\label{eq:optimal_error}
    p_{\rm err}^{\rm opt} = \frac{1}{2}(1-{\rm TV}[p,q]),
\end{equation}
giving the TV distance operational meaning. In Ref.~\cite{arnold:2023} (Sec. S2 in corresponding Supplemental Material), it was explicitly shown that the optimal predictive model in approach 2 for detecting phase transitions from data discussed in the main text makes predictions according to the Neyman-Pearson strategy and achieves the optimal error rate in Eq.~\eqref{eq:optimal_error}. 

\subsection{Parameter estimation}
In the statistical inference task of parameter estimation, one tries to estimate an unknown set of parameters $\bm{\gamma}$ based on independent measurements $\bm{x}$ distributed according to the probability distribution $p_{\bm{\gamma}}(\bm{x})$. This is done by calculating an estimator of $\bm{\gamma}$ denoted $\hat{\bm{\gamma}}(\bm{x})$ given a measurement $\bm{x}$. The bias of an estimator is defined as
\begin{equation}
    \bm{b}(\bm{\gamma}) = \mathbb{E}_{\bm{x} \sim p_{\bm{\gamma}}}[\hat{\bm{\gamma}}(\bm{x}) - \bm{\gamma}] = \bm{\psi}(\bm{\gamma}) - \bm{\gamma}.
\end{equation}
An estimator is called \emph{unbiased} if $\bm{b}(\bm{\gamma}) = 0\; \forall \bm{\gamma}$. The \emph{Cramér-Rao bound} (also known as information inequality) states that
\begin{equation}\label{cr:bound}
    {\rm Cov}(\bm{\gamma}) \geq \left(\frac{\partial \bm{\psi}}{\partial \bm{\gamma}}\right)^{T} \left(\mathcal{F}(\bm{\gamma})\right)^{-1} \left(\frac{\partial \bm{\psi}}{\partial \bm{\gamma}}\right),
\end{equation}
where ${\rm Cov}(\bm{\gamma})$ is the covariance matrix of the estimator evaluated at $\bm{\gamma}$ and $\left(\frac{\partial \bm{\psi}}{\partial \bm{\gamma}}\right)$ denotes the Jacobian matrix of $\bm{\psi}$ with matrix elements $\left(\frac{\partial \bm{\psi}}{\partial \bm{\gamma}}\right)_{i,j} = \partial \psi_{i}(\bm{\gamma})/\partial \gamma_{j}$. Here, we have assumed that $\mathcal{F}(\bm{\gamma})$ is nonsingular. A proof can be found in Ref.~\cite{bickel:2015}, pp. 179–188.\\

For an unbiased estimator (i.e.,  $\bm{\psi}(\bm{\gamma}) = \bm{\gamma}$), the Cramér-Rao bound reduces to   
\begin{equation}\label{cr:bound2}
    {\rm Cov}(\bm{\gamma}) \geq \left(\mathcal{F}(\bm{\gamma})\right)^{-1} ,
\end{equation}
In the special scalar case, the Cramér-Rao bound reads
\begin{equation}\label{cr:bound}
    \sigma^2(\gamma) \geq \frac{[1+b'(\gamma)]^2}{\mathcal{F}(\gamma)} = \frac{(\psi'(\gamma))^2}{\mathcal{F}(\gamma)},
\end{equation}
where $b'(\gamma) = \frac{\partial b}{\partial \gamma}$ and $\mathcal{F}$ the Fisher information associated with $\gamma$.\\

\section{Generalization to higher-dimensional parameter spaces}
In this section, we generalize the results of the main text to parameter spaces of arbitrary dimension which possibly feature multiple phase transitions. To this end, we assume our physical system to be characterized by a vector $\bm{\gamma}$ of tuning parameters. Measurements (samples) of the system state are obtained at a discrete set of realizations of the tuning parameter $\Gamma$. Generalizations of the three approaches for detecting phase transitions from data discussed in the main text to higher-dimensional parameter spaces have been proposed in Ref.~\cite{arnold:2023}. In what follows, we will show that the corresponding indicators are lower bounds to the trace of the Fisher information matrix.\\

For approach 1, we assume knowledge of the number of distinct phases, $K$, and their rough location in parameter space. The corresponding points in parameter space are assigned distinct labels $y \in \mathcal{Y} = \{1,\dots,K \}$. The indicator is then given by 
\begin{equation}\label{eq:ISL}
    I_{1}(\bm{\gamma}) = \frac{1}{K} \sum_{y \in \mathcal{Y}} \sqrt{ \sum_{i=1}^d\left(\frac{\partial Q(y|\bm{\gamma})}{\partial \gamma_{i}}\right)^2 },
\end{equation}
where $Q(y|\bm{\gamma}) = \mathbb{E}_{\bm{x} \sim P(\cdot|\bm{\gamma})}\left[ Q(y|\bm{x}) \right]$ with $Q(y|\bm{x})$ denoting the (estimated) probability that sample $\bm{x}$ carries label $y$. Following the derivation of the main text for the conditional probability of each label $Q(y|\bm{\gamma})$ separately, we have
\begin{equation}
    I_{1}(\bm{\gamma}) \leq \frac{1}{K} \sum_{y \in \mathcal{Y}} \sigma_{y}(\bm{\gamma}) \sqrt{ \sum_{i=1}^{d} \mathcal{F}_{i,i}(\bm{\gamma}) } = \sqrt{ {\rm tr}\left(\mathcal{F}(\bm{\gamma})\right) },
\end{equation}
where $\mathcal{F}_{i,i}(\bm{\gamma})$ is the $i$th diagonal element of the Fisher information matrix and $\sigma_{y}(\bm{\gamma})$ is the standard deviation of the $Q(y|\bm{x})$ at $\bm{\gamma}$.
\\

\indent Approach 2 can be generalized by dividing the parameter space along each direction at each sampled point $\bm{\gamma}= \left(\gamma_{1}, \gamma_{2},\dots,\gamma_{d}\right) \in \Gamma$. For a given direction $1 \leq i \leq d$, this yields two sets, $\Gamma_{1}^{(i)}(\bm{\gamma})$ and $\Gamma_{2}^{(i)}(\bm{\gamma})$, each comprised of the $l$ points closest to $\bm{\gamma}$ in part 1 and 2 of the split parameter space, respectively. Based on these sets, an indicator component is computed according to Eqs. (2) and (3) in the main text, $I_{2}^{(i)}(\bm{\gamma}) = 1 - 2 p_{\rm err}^{(i)}(\bm{\gamma})$. The overall indicator is
\begin{equation}\label{eq:ILBC}
    I_{2}(\bm{\gamma}) = \sqrt{\sum_{i=1}^{d} \left(I_{2}^{(i)}(\bm{\gamma})\right)^2 }.
\end{equation}
Following the proof in the main text, for $l=1$ each indicator component is bounded as
\begin{equation}
    I_{2}^{(i)}(\bm{\gamma}) \leq I_{2}^{(i),\rm opt}(\bm{\gamma}) \leq \frac{1}{2}\sqrt{\mathcal{F}_{i,i}(\bm{\gamma})}\delta\gamma_{i} + \mathcal{O}(\delta\gamma_{i}^2).
\end{equation}
Thus,
\begin{equation}\label{eq:ILBC}
    I_{2}(\bm{\gamma}) \leq I_{2}^{\rm opt}(\bm{\gamma}) \leq \sqrt{\sum_{i=1}^{d} \left(I_{2}^{(i),{\rm opt}}(\bm{\gamma})\right)^2 }\leq \frac{1}{2}\sqrt{ \sum_{i=1}^{d}{\mathcal{F}_{i,i}(\bm{\gamma})}\delta\gamma_{i}^2} + \sum_{i=1}^{d} \mathcal{O}(\delta\gamma_{i}^2).
\end{equation}
Assuming $\delta\gamma_{i} = \delta\gamma \;\forall i$, we have
\begin{equation}\label{eq:ILBC}
    I_{2}(\bm{\gamma}) \leq I_{2}^{\rm opt}(\bm{\gamma}) \leq \frac{1}{2}\delta\gamma\sqrt{ {\rm tr}\left({\mathcal{F}(\bm{\gamma})}\right)} +  \mathcal{O}(\delta\gamma^2).
\end{equation}
Thus, in the limit $\delta \gamma \rightarrow 0$, the quantity $2I_{2}(\bm{\gamma})/\delta \gamma$ serves as a lower bound to the trace of the Fisher information matrix.\\

\indent The indicator of approach 3 can be generalized as
\begin{equation}\label{eq:I3}
    I_{3}(\bm{\gamma}) = \sqrt{\sum_{i=1}^{d} \left(\frac{\partial \hat{\gamma}_{i}(\bm{\gamma}) / \partial \gamma_{i} }{\sigma_{i}(\bm{\gamma})}\right)^2},
\end{equation}
where $\hat{\bm{\gamma}}(\bm{\gamma}) = \mathbb{E}_{\bm{x} \sim P(\cdot|\bm{\gamma})}\left[ \hat{\bm{\gamma}}(\bm{x})\right]$ and $\bm{\sigma}(\bm{\gamma}) = \sqrt{\mathbb{E}_{\bm{x} \sim P(\cdot|\bm{\gamma})}\left[ \hat{\bm{\gamma}}(\bm{x})^2\right] - \left(\mathbb{E}_{\bm{x} \sim P(\cdot|\bm{\gamma})}\left[ \hat{\bm{\gamma}}(\bm{x})\right]\right)^2}$ with $\hat{\bm{\gamma}}(\bm{x})$ being an estimator for $\bm{\gamma}$ (operations are carried out element-wise). Applying the scalar Cramér-Rao bound [Eq.~\eqref{cr:bound}] to each component of the sum in Eq.~\eqref{eq:I3}, we obtain the bound 
\begin{equation}
    I_{3}(\bm{\gamma}) \leq \sqrt{{\rm tr}\left(\mathcal{F}(\bm{\gamma})\right)}.
\end{equation}

\section{Additional results and insights}
\subsection{Approach 1: Choice of training regions and improved bounds}
The indicator of approach 1 (including its peak position) depends heavily on the choice of $\Gamma_{0}$ and $\Gamma_{1}$, i.e., on prior knowledge of the location of the phase transition. Figure~\ref{fig:SI_1} shows optimal indicators $I_{1}^{\rm opt}$ for various choices of the training regions in the case of the two-dimensional classical Ising model. Note that while the peak position shifts depending on the choice of training regions, the largest indicator value (i.e., the best bound to the Fisher information) is achieved when the parameter space is split at the critical point. This suggests that, in practice, the indicator should be computed for various choices of the sets $\Gamma_{0}$ and $\Gamma_{1}$. By tracing the maximum indicator value attained at each point in parameter space, an overall improved lower bound to the square root of the Fisher information can be obtained, see envelopes given by dashed and dotted lines in Fig.~\ref{fig:SI_1}.\\

\indent In the main text, we pointed out that approach 1 could be improved by modifying the indicator $I_1(\gamma) \mapsto \tilde{I}_1(\gamma) = I_1(\gamma)/\sigma(\gamma)$. Figure~\ref{fig:SI_1} also displays the rescaled optimal indicators $\tilde{I}_{1}^{\rm opt}$ for various choices of the training regions, highlighting that this small modification indeed yields improved bounds and, in fact, accurate approximations to the Fisher information near the bipartition boundary.\\

\begin{figure}[tbh!]
	\centering
		\includegraphics[width=0.7\linewidth]{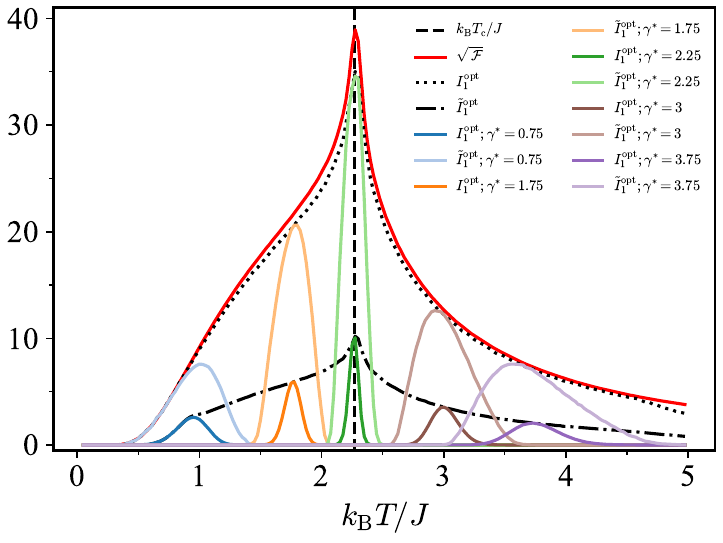}
		\caption{Results of approach 1 applied to the square-lattice ferromagnetic Ising model ($L=60$) with tuning parameter $\gamma = k_{\rm B}T/J$. The critical point $k_{\rm B}T_{\rm c}/J = 2 / \ln(1+\sqrt{2})$ is highlighted by a vertical black-dashed line. All shown quantities are lower bounds to the square root of the system's Fisher information, here corresponding to $CJ^2/k_{\rm B}^3T^2$. The indicators $I_{1}^{\rm opt}$ and $\tilde{I}_{1}^{\rm opt}$  are computed for bipartitions of the parameter range $\Gamma = \{0.025,0.1,\dots,5.0 \}$ into two regions $\Gamma_{0} =\{0.025,\dots,\gamma^{*} \}$ and $\Gamma_{1} = \{ \gamma^{*}+0.025,\dots,5.0 \}$. For both $I_{1}^{\rm opt}$ and $\tilde{I}_{1}^{\rm opt}$, the five bold lines correspond to five distinct $\gamma^{*} \in \{ 0.75,1.75,2.25,3,3.75 \}$ that determine the location of the bipartition. Note that the bipartition at $2.25$ almost coincides with the critical point located at $\approx 2.27$. The envelopes of the two indicator curves across all possible biparitions, i.e., all possible choices of $\gamma^{*}$, are shown by dashed and dotted lines. The set $\Gamma$ is composed of a uniform grid with 200 points ranging from $\gamma = 0.025$ to $\gamma = 5$ (grid spacing $\delta \gamma = 0.025$). Each dataset $D_{\gamma}$ consists of $10^6$ spin configurations.}
		\label{fig:SI_1}
\end{figure}

\indent Let us discuss a natural scenario in which the optimal predictive model $\hat{y}^{\rm opt}$ with respect to the cross-entropy loss [Eq. (5) in the main text] is strongly correlated with the score $\frac{\partial \log P(\cdot|\gamma)}{\partial \gamma}$. For this, let $\Gamma_0 = \lbrace \gamma_0 \rbrace$ and $\Gamma_{1} = \lbrace \gamma_1 \rbrace$ be singletons such that $\gamma_0 < \gamma_1$. In this case, $\hat{y}^{\rm opt}$ admits a closed-form expression~\cite{arnold:2022,arnold:2023}
\begin{align*}
    \hat{y}^{\rm opt}(\bm{x}) = 1-\frac{P(\bm{x}|\gamma_0)}{P(\bm{x}|\gamma_0) + P(\bm{x}|\gamma_1)}.
\end{align*}
When $\gamma_0$ and $\gamma_1$ are far away from each other, there is, at first glance, no reason to expect that $\hat{y}^{\rm opt}$ will be well correlated with $\frac{\partial \log P(\cdot|\gamma)}{\partial \gamma}$ for any $\gamma \in [\gamma_0, \gamma_1]$. However, if we consider $\gamma_0$ and $\gamma_1$ to be close, i.e., $\gamma_1 = \gamma_0 + \delta \gamma$, and assume that the distributions evolve sufficiently smoothly with $\gamma$ so that $\left.\frac{\partial P(\cdot|\gamma)}{\partial \gamma}\right|_{\gamma = \gamma_{0}} \delta \gamma$ is a valid first-order approximation for $\delta P = P(\cdot |\gamma_1) - P(\cdot | \gamma_0)$, we can note that 
\begin{align*}
    \hat{y}^{\rm opt}(\bm{x}) = 1-\frac{1}{2} \frac{1}{1 + \frac{\delta P(\bm{x})}{2P(\bm{x}|\gamma_0)}} = \frac{1}{2}  + \frac{1}{4} \frac{\delta P(\bm{x})}{P(\bm{x}|\gamma_0)} + \mathcal{O}\left(\delta P(\bm{x})^2  \right) =  \frac{1}{2} + \frac{1}{4}\left.\frac{\partial \log P(\bm{x}|\gamma)}{ \partial \gamma}\right|_{\gamma = \gamma_0} \delta \gamma  + \mathcal{O}\left(\delta \gamma^2\right).
\end{align*}
Thus, in the limit of small $\delta \gamma$, $\hat{y}$ is indeed strongly correlated with $\left. \frac{\partial \log P(\cdot |\gamma)}{\partial \gamma}\right|_{\gamma_0}$ and so $I_1(\gamma_0)/\delta \gamma$ is a first-order accurate, constant factor approximation for $\mathcal{F}(\gamma_0)$, $I_1(\gamma_0)/\delta \gamma \to \frac{1}{4} \mathcal{F}(\gamma_0)$ as $\delta \gamma \to 0$ (note that $I_1$ depends implicitly on $\delta \gamma$). This scenario is similar to the one of approach 2 described in the main text ($l=1$) and suggests a new scheme for detecting phase transitions by approximating the Fisher information in a data-driven manner:  
for each $\gamma \in \Gamma$ where $\Gamma$ is composed of a uniform grid with spacing $\delta \gamma$, let $\gamma_{0} = \gamma$ and $\gamma_{1} = \gamma + \delta \gamma$ and obtain a point-wise estimate of the Fisher information as $4 I_{1}(\gamma)/\delta \gamma$.

\subsection{Approach 2: Neural network-based results and alternative bound}
In the main text, we proved that the indicator of approach 2 arising from any predictive model serves as a lower bound to the indicator of the corresponding Bayes-optimal predictive model. Figure~\ref{fig:3}(a) shows the indicator of approach 2 for a neural network-based (NN-based) predictive model at various stages of training. The NN-based indicator quickly approaches the optimal indicator, highlighting that good approximations to the Fisher information may also be achieved in practice via explicit data-driven training.\\

It has been known that the loss function of approach 2 also carries information about the underlying phase transition and may serve as an alternative indicator function~\cite{arnold:2022}. In the following, we will put this intuition on firm footing:\\ 

In the infinite data limit, the loss function of approach 2 [Eq. (5) in the main text] becomes 
\begin{equation}
    \mathcal{L}_{\rm 2} = - \frac{1}{2} \left(\mathbb{E}_{\bm{x} \sim P_{0}}\left[\ln \left( 1-\hat{y}_{\bm{\theta}}(\bm{x})\right)\right]  + \mathbb{E}_{\bm{x} \sim P_{1}}\left[\ln \left( \hat{y}_{\bm{\theta}}(\bm{x})\right)\right]  \right)
\end{equation}
where $P_{y} = \frac{1}{|\Gamma_{y}|} \sum_{\gamma' \in \Gamma_{y}} P(\cdot|\gamma'), \; y \in \{ 0,1\}$. By definition $\mathcal{L}_{2} \geq \mathcal{L}_{2}^{\rm opt}$ where $\mathcal{L}_{2}^{\rm opt}$ is the global minimum of the loss function attained by the Bayes-optimal strategy. In approach 2, the optimal model makes the following predictions $\hat{y}^{\rm opt}(\bm{x}) = \frac{P_{1}(\bm{x})}{P_{0}(\bm{x}) + P_{1}(\bm{x})}$ (see Sec.~\ref{sec:HT}). Thus, 
\begin{equation}
    \mathcal{L}_{2}^{\rm opt} = - \frac{1}{2} \left(\mathbb{E}_{\bm{x} \sim P_{0}}\left[\ln \left( \frac{P_{0}(\bm{x})}{P_{0}(\bm{x}) + P_{1}(\bm{x})}\right)\right]  + \mathbb{E}_{\bm{x} \sim P_{1}}\left[\ln \left( \frac{P_{1}(\bm{x})}{P_{0}(\bm{x}) + P_{1}(\bm{x})} \right)\right]  \right) = -{\rm JS}\left[P_{0},P_{1}\right] + \ln(2).
\end{equation}
Hence, ${\rm JS}\left[P_{0},P_{1}\right] = \ln(2) - \mathcal{L}_{2}^{\rm opt} \geq \ln(2) - \mathcal{L}_{2}$. Next, we choose $l=1$ and the two sets of points $\Gamma_{0}$ and $\Gamma_{1}$ such that $P_{0}$ and $P_{1}$ correspond to probability distributions separated by a small distance $\delta \gamma$ in parameter space. Expanding the Jensen-Shannon divergence to lowest order according to $D_{f}[p_{\bm{\gamma}},p_{\bm{\gamma}+\bm{\delta \gamma}}] = \frac{f''(1)}{2}\bm{\delta \gamma}^T \mathcal{F}(\bm{\gamma})\bm{\delta \gamma} + \mathcal{O}(\delta \gamma^3)$ with $f''(1) = 1/4$ (note that $f'(1)=0$), we have ${\rm JS}\left[P_{0},P_{1}\right] = \frac{\delta \gamma^2}{8}\mathcal{F}(\gamma) + \mathcal{O}(\delta \gamma^3)$. Together with the above bound, this yields 
\begin{equation}\label{eq:SI_L_bound}
    8(\ln(2) - \mathcal{L}_{2})/\delta \gamma^2 \leq 8(\ln(2) - \mathcal{L}_{2}^{\rm opt})/\delta \gamma^2 = \mathcal{F}(\gamma) + \mathcal{O}(\delta \gamma).
\end{equation}
That is, an affine transformation of the loss value serves as a lower bound to the Fisher information in the limit $\delta \gamma \rightarrow 0$.\\

Note that this bound is based on relating the optimal loss value to the Jensen-Shannon divergence [Eq.~\eqref{eq:JSD}] which is an $f$-divergence. An expansion in lowest order thus yields a tie to the Fisher information (recall the discussion in Sec.~\ref{sec:inf_theory}). As such, this corresponds to another data-driven scheme for estimating the Fisher information based on approximating an $f$-divergence. The results for the Ising model as shown in Fig.~\ref{fig:3}(b).

\begin{figure}[tbh!]
	\centering
		\includegraphics[width=0.9\linewidth]{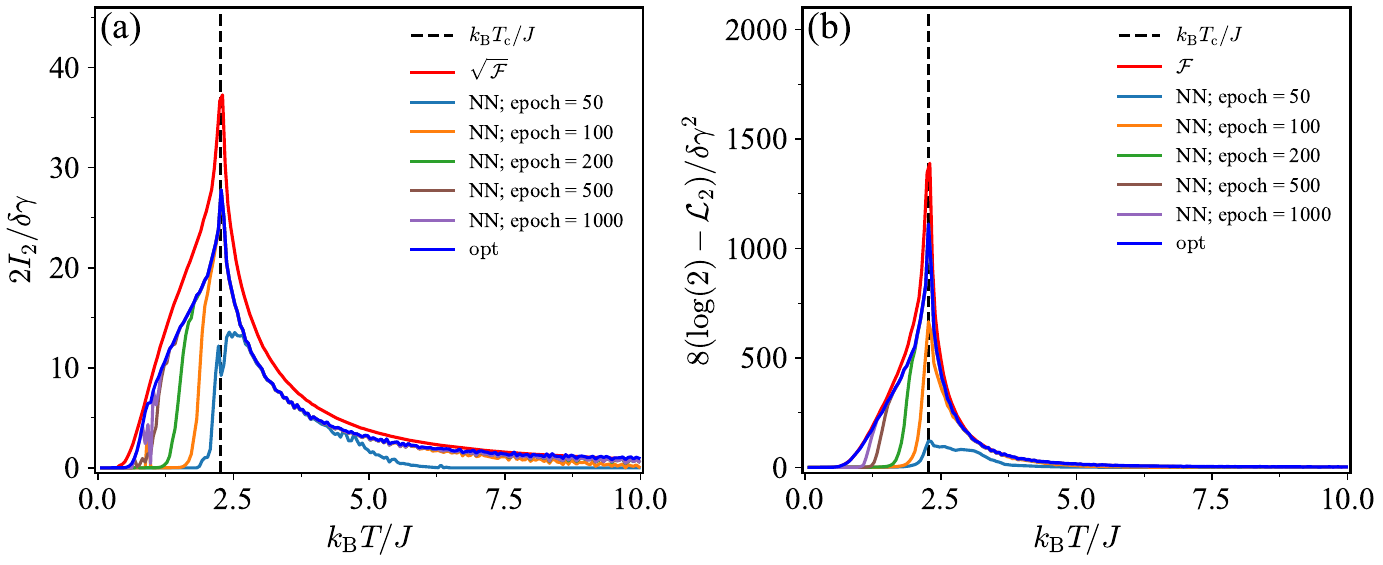}
		\caption{Results of approach 2 applied to the square-lattice ferromagnetic Ising model ($L=60$) with tuning parameter $\gamma = k_{\rm B}T/J$. The critical point $k_{\rm B}T_{\rm c}/J = 2/ \ln(1+\sqrt{2})$ is highlighted by a vertical black-dashed line. The results of an NN-based classifier for various training epochs are shown in green. Indicators corresponding to the (approximate) Bayes-optimal classifier are shown in blue~\cite{arnold:2022,arnold:2023}. (a) Rescaled indicator $I_{2}$ [Eq.~(3) in the main text, $l=1$]. (b) Bound based on loss function [Eq.~\eqref{eq:SI_L_bound}]. The set $\Gamma$ is composed of a uniform grid with 200 points ranging from $\gamma = 0.05$ to $\gamma = 10$ (grid spacing $\delta \gamma = 0.05$). Each dataset $D_{\gamma}$ consists of $10^5$ spin configurations. We consider feedforward NNs (implemented using Flux~\cite{innes:2018} in \texttt{Julia}~\cite{bezanson:2012}) with three hidden layers composed of 64 nodes each, rectified linear units as activation functions, and a learning rate of $5 \times 10^{-4}$. Weights and biases are optimized via gradient descent with Adam~\cite{kingma:2014}, where the gradients are calculated using backpropagation. As an NN input, we use the energy of a sample, which corresponds to the sufficient statistic~\cite{arnold:2023}. The inputs are standardized before training.}
		\label{fig:3}
\end{figure}

\subsection{Approach 3: Bias-variance tradeoff}
The loss function of approach 3 [mean squared error in Eq. (6) in the main text] 
\begin{equation}
    \mathcal{L}_{\rm 2} = \frac{1}{|\Gamma|}\sum_{\gamma \in \Gamma} \mathbb{E}_{\bm{x} \sim P(\cdot|\gamma)}[(\hat{\gamma}(\bm{x}) - \gamma)^2],
\end{equation}
can be decomposed into variance $\sigma^2(\gamma) = \mathbb{E}_{\bm{x} \sim P(\cdot|\gamma)} \left[\left(\hat{\gamma}(\bm{x}) - \mathbb{E}_{\bm{x} \sim P(\cdot|\gamma)}(\hat{\gamma}(\bm{x})) \right)^2\right]$ and bias terms $b(\gamma) = \mathbb{E}_{\bm{x} \sim P(\cdot|\gamma)} \left[\hat{\gamma}(\bm{x}) - \gamma \right]$:
\begin{equation}
    \mathcal{L}_{\rm 2} = \frac{1}{|\Gamma|}\sum_{\gamma \in \Gamma} \sigma^2(\gamma) + b^2(\gamma) = \langle \sigma^2\rangle + \langle b^2 \rangle,
\end{equation}
given that
\begin{equation}
    \mathbb{E}_{\bm{x} \sim P(\cdot|\gamma)}\left[(\hat{\gamma}(\bm{x}) - \gamma)^2\right] = \sigma^2(\gamma) + b^2(\gamma).
\end{equation}
When working with a finite dataset, the variance and bias should be replaced by their finite sample approximations (i.e., replacing expected values with a sample mean $\mathbb{E}_{\bm{x} \sim P(\cdot|\gamma)}\left[\cdot\right] \mapsto \frac{1}{|\mathcal{D}_{\gamma}|}\sum_{\bm{x} \in \mathcal{D}_{\gamma}}\left[\cdot\right]$). 

The indicator of approach 3 
\begin{equation}\label{eq:IPBM_SI}
     I_{3}(\gamma) = \frac{\partial \hat{\gamma}(\gamma) / \partial \gamma }{\sigma(\gamma)} =  \frac{\partial b(\gamma)/\partial \gamma + 1 }{\sigma(\gamma)}, 
\end{equation}
where $\hat{\gamma}(\gamma) = \mathbb{E}_{\bm{x} \sim P(\cdot|\gamma)}\left[ \hat{\gamma}(\bm{x}) \right]$, can be viewed as a ``signal-to-noise'' ratio where the ``signal'' term $\partial \hat{\gamma}(\gamma) / \partial \gamma$ corresponds to the change in the bias of the estimator and the ``noise'' term corresponds to the standard deviation of the estimator. Both have independently been used as indicators of phase transitions. In particular, it has been noted that the change in the bias of the estimator alone is a reliable indicator early on during training of NN-based predictive models~\cite{schaefer:2019,arnold:2022}. However, as the capacity of the predictive models increases, the change in their bias can show additional, erroneous peaks that do not correspond to critical points~\cite{schaefer:2019,arnold:2022}. This is illustrated in Fig.~\ref{fig:4}(b) in the case of the Ising model. These erroneous peaks can be removed by dividing by the standard deviation~\cite{arnold:2023}. Due to the bias-variance tradeoff~\cite{friedman:2001}, during NN training, the bias contribution to the loss typically decreases while the variance contribution increases. This is shown for the Ising model in Fig.~\ref{fig:4}(d). If the decrease of the bias precedes the increase in variance, it is expected that the change in the bias alone can constitute a reliable indicator early on during training. In contrast, the standard deviation is expected to be reliable only at later stages, see Fig.~\ref{fig:4}(d). Their ratio, the indicator proposed in Eq.~\eqref{eq:IPBM_SI}, yields a reliable signal throughout both stages of training, see Fig.~\ref{fig:4}(a). In particular, their ratio yields a fairly good lower bound to the square root of the Fisher information even at the early stages of training.

\begin{figure}[tbh!]
	\centering
		\includegraphics[width=0.99\linewidth]{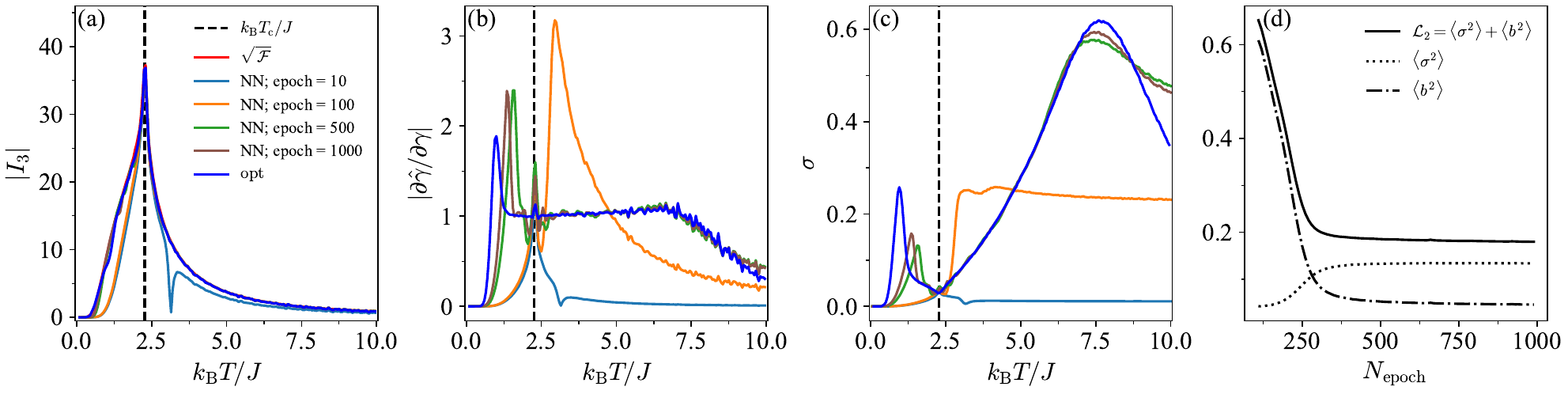}
		\caption{Results of approach 3 applied to the square-lattice ferromagnetic Ising model ($L=60$) with tuning parameter $\gamma = k_{\rm B}T/J$. The critical point $k_{\rm B}T_{\rm c}/J = 2/ \ln(1+\sqrt{2})$ is highlighted by a vertical black-dashed line. (a) The optimal indicator (blue) as well as NN-based indicators (green) are lower bounds to the square root of the system's Fisher information (red). (b) Numerator of the indicator $I_{3}$ [Eq.~(4) in the main text] in optimal case (blue) and NN-based case after various training epochs (green). (c) Denominator of indicator $I_{3}$ [Eq.~(4) in the main text] in optimal case (blue) and NN-based case after various training epochs (green). (d) Loss function of NN-based predictive model and its bias-variance decomposition as a function of the number of training epochs. The set $\Gamma$ is composed of a uniform grid with 200 points ranging from $\gamma = 0.05$ to $\gamma = 10$ (grid spacing $\delta \gamma = 0.05$). Each dataset $D_{\gamma}$ consists of $10^5$ spin configurations. We consider feedforward NNs (implemented using Flux~\cite{innes:2018} in \texttt{Julia}~\cite{bezanson:2012}) with three hidden layers with 64 nodes each, rectified linear units as activation functions, and a learning rate of $10^{-3}$. Weights and biases are optimized via gradient descent with Adam~\cite{kingma:2014}, where the gradients are calculated using backpropagation. As an NN input, we use the energy of a sample, which corresponds to the sufficient statistic~\cite{arnold:2023}. The inputs are standardized before training.}
		\label{fig:4}
\end{figure}

\subsection{Relation between generative adversarial network fidelity and Fisher information}
In this section, we will show that the generative adversarial network (GAN) fidelity which has been proposed in Ref.~\cite{singh:2021} as an indicator of phase transitions is also related to the square root of the system's Fisher information. The GAN fidelity is defined as
\begin{equation}
    F_{\rm GAN}(\gamma) = \frac{1}{\delta \gamma} \mathbb{E}_{z \sim p}\left[D\left(G(z|\gamma),\gamma\right) - D\left(G(z|\gamma),\gamma + \delta \gamma\right)\right].
\end{equation}
The function $D(\bm{x},\gamma)$ is a \emph{discriminator} trained to output 1 for samples $\bm{x}$ drawn from the probability distribution $P(\cdot|\gamma)$ and 0 otherwise, whereas the function $G(\cdot|\gamma): \mathcal{Z} \rightarrow \mathcal{X}$ is a \emph{generator} trained to produce samples from $P(\cdot|\gamma)$ when evaluated with $z \sim p$, where $p$ is a simple prior distribution over $z \in \mathcal{Z}$. We are going to analyze the ideal case in which both the generator and discriminator are optimal, i.e., implement optimal strategies for generating and discriminating samples, respectively. In this case, we have
\begin{equation}
    F_{\rm GAN}^{\rm opt}(\gamma) = \frac{1}{\delta \gamma} \mathbb{E}_{\bm{x} \sim P(\cdot|\gamma)}\left[D^{\rm opt}\left(\bm{x},\gamma\right) - D^{\rm opt}\left(\bm{x},\gamma + \delta \gamma\right)\right],
\end{equation}
with $D^{\rm opt}\left(\bm{x},\gamma\right) = P(\bm{x}|\gamma)/\mathcal{N}(\bm{x})$ where $\mathcal{N}(\bm{x}) = \sum_{\gamma \in \Gamma} P(\bm{x}|\gamma)$ is a normalization factor. Taking the limit $\delta \gamma \rightarrow 0$, we have
\begin{equation}
    F_{\rm GAN}^{\rm opt}(\gamma) = - \mathbb{E}_{\bm{x} \sim P(\cdot|\gamma)}\left[\frac{1}{ \mathcal{N}(\bm{x})}  \frac{\partial P(\bm{x}|\gamma)}{\partial \gamma}\right] \leq \sum_{\bm{x} \in \mathcal{X} } D^{\rm opt}(\bm{x},\gamma)\left| \frac{\partial P(\bm{x}|\gamma)}{\partial \gamma}  \right| \leq \sum_{\bm{x} \in \mathcal{X} } \left| \frac{\partial P(\bm{x}|\gamma)}{\partial \gamma}\right|  \leq\sqrt{\mathcal{F}} ,
\end{equation}
where we have used the fact that $0 \leq D^{\rm opt}(\bm{x},\gamma) \leq 1$ for the second inequality and the Cauchy-Schwarz inequality for the last step.
\section{Data generation}
\subsection{Ising model}
Given a fixed set of tuning parameter values, we use the Metropolis-Hastings algorithm to sample spin configurations from the corresponding Boltzmann distribution. We initialize the system in one of its two lowest energy states. The lattice is updated by drawing a random spin and flipping it. The new state is accepted with probability ${\rm min}(1, e^{-\Delta E/k_{\rm B}T})$, where $\Delta E$ is the energy difference resulting from the spin flip. After a thermalization period of $N$ lattice sweeps, we collect $N$ samples, where we set $N = 10^5$ or $N = 10^6$ as specified in the respective figure captions.

\subsection{Transverse-field Ising model}
To obtain ground states of the one-dimensional quantum transverse-field Ising model, we perform exact diagonalization using the QuSpin package~\cite{quspin:2017,quspin:2019} in \texttt{Python}. We compute the quantum Fisher information as the second derivative of the fidelity $ F(\rho,\sigma) = {\rm tr}(\sqrt
{\sqrt{\sigma}\rho\sqrt{\sigma}})$, where $F(\rho(\gamma),\rho(\gamma+\delta \gamma)) = 1 - \delta \gamma^2 \mathcal{F}^{Q}(\rho(\gamma))/8 + \mathcal{O}(\delta \gamma^3)$. We compute a finite difference approximation of the second derivative using the second-order central difference formula.

\end{document}